\documentclass[
  a4paper,
  onecolumn,
  titlepage,
  11pt,
  unpublished,
  amsmath,
  amssymb
]{quantumarticle}
\pdfoutput=1

\usepackage{amsthm}
\usepackage{booktabs}
\usepackage{graphicx}
\usepackage{algorithm}
\usepackage{algpseudocode}
\usepackage{placeins}
\usepackage[numbers,sort&compress]{natbib}
\usepackage{xurl}
\usepackage{tikz}
\usetikzlibrary{arrows.meta,positioning,fit,backgrounds,calc}
\usepackage{hyperref}

\AtBeginDocument{%
  \hypersetup{%
    colorlinks=true,
    linkcolor=[RGB]{200,90,20},
    citecolor=[RGB]{0,90,160},
    urlcolor=[RGB]{0,90,160}%
  }%
}

\newcommand{\code}[1]{\path{#1}}
\newcommand{\ket}[1]{\ensuremath{|#1\rangle}}
\newcommand{\bra}[1]{\ensuremath{\langle#1|}}

\newtheorem{lemma}{Lemma}
\newtheorem{proposition}{Proposition}
\newtheorem{corollary}{Corollary}
\newtheorem{remark}{Remark}

\begin{document}

\title{Pauli Decomposition by Character Theory: A Memory-Bounded
Algorithm for Qubits and Qudits}

\author{Mohammadreza Khellat}
\email{mkhellat@beavernets.com}
\orcid{0000-0002-8236-809X}
\affiliation{Beavernets Technologies}

\author{Mohammad Masoumi}
\orcid{0000-0002-6955-834X}
\affiliation{Beavernets Technologies}

\author{Saman Nasoori}
\affiliation{Beavernets Technologies}

\author{Soroush Nasoori}
\affiliation{Beavernets Technologies}

\begin{abstract}
Pauli decomposition of an $n$-qubit operator is expensive if each of
the $4^n$ Pauli strings is treated as an unstructured object. Labeling
those strings by the configuration space $(\mathbb{Z}_2)^n$ makes them
displacement operators $X^{\otimes x}Z^{\otimes z}$ on the
computational space: shifts by~$x$ and multiplications by characters
$Z^{\otimes z}$ in the computational basis. That labeling organizes
the operator into $2^n$ fibers of length $2^n$, one per shift~$x$,
each collecting the computational-basis amplitudes that hop by~$x$.
Expanding each fiber function in the character basis of
$(\mathbb{Z}_2)^n$ --- a length-$2^n$ Walsh--Hadamard transform ---
recovers the Pauli coefficients at total cost $O(n\cdot4^n)$. The same
fiber /
shift--character picture extends to $(\mathbb{Z}_p)^n$
(Heisenberg--Weyl qudits). \texttt{paulikit} is a memory-bounded
realization for qubits: for matrix-sparse inputs, peak memory need not
densify to a full $2^n\times2^n$ operator. Performance is checked on a
stated low-end GNU/Linux platform for a coupled-oscillator sweep
beyond a billion Pauli terms.
\end{abstract}

\maketitle

\section{Introduction}
\label{sec:intro}

Simulating an operator $T$ on a gate-model quantum computer --- a
Hamiltonian in the usual case --- presupposes its availability to
termwise circuit action in a Pauli basis the hardware can apply.
Pauli decomposition makes a simulated operator available to that
termwise action as
\begin{equation}
T=\sum_{x,z}d_{x,z}\,P(x,z),
\label{eq:intro-pauli-decomposition}
\end{equation}
with $P(x,z)$ ranging over the $n$-qubit Pauli strings (tensor
products of $\{I,X,Y,Z\}$) and $d_{x,z}\in\mathbb{C}$ the
corresponding coefficients. Trotterization~\citep{abrams_lloyd_1997}
consumes such a sum term by term in a product formula; phase
estimation~\citep{abrams_lloyd_1999,aspuru_guzik_2005} uses simulated
evolution built from the same expansion --- including molecular
ground-state calculations that run phase estimation on a chemically
mapped Hamiltonian~\citep{aspuru_guzik_2005}. The sparse-Hamiltonian
step inside HHL, when that operator is itself a Pauli
sum~\citep{harrow_hassidim_lloyd_2009}, and variational ans\"atze built
from Pauli observables, inherit the same coefficient cost. The subject
of the present work
is that cost --- the gap between a naive $4^n$-scale treatment and the
fast Walsh--Hadamard route --- together with the range of the fast
argument. The presentation is for qubits; the same character-theoretic
construction specializes to qudits
(Corollary~\ref{cor:qudit-decomposition}). In \texttt{paulikit},
non-Hermitian operators are already
a first-class tested input (\texttt{assume\_hermitian=False}), while a
qudit module is not yet shipped. A worked non-Hermitian physical
companion beside the Hermitian oscillator family of
Section~\ref{sec:motivation}, and a deployed qudit module on
Corollary~\ref{cor:qudit-decomposition}, are left to follow-on work.

\subsection{Contributions}
\label{sec:contributions-skim}

The contributions may be summarized as follows.
(i)~Pauli decomposition is derived from the character theory of
$(\mathbb{Z}_2)^n$, reducing to $2^n$ independent Walsh--Hadamard
transforms at cost $O(n\cdot 4^n)$; the same argument yields the qudit
case via the DFT over $(\mathbb{Z}_p)^n$
(Corollary~\ref{cor:qudit-decomposition}).
(ii)~The character-sum / Walsh--Hadamard stage already produces a
complete, cost-equivalent decomposition in the bare Pauli basis;
passage to the genuine basis with $Y=iXZ$ is a phase convention rather
than an algebraic requirement (Remark~\ref{rem:why-genuine}).
(iii)~Relative to Georges et al.\ \citep{georges_et_al_2025}, the
asymptotic time class is unchanged. What is put first is peak memory
independent of densifying a matrix-sparse input to a full
$2^n\times 2^n$ operator, as opposed to an $O(1)$ additional working
set on a mandatory dense buffer (an already-dense ndarray input still
occupies its own $\dim\times\dim$ footprint). The claim is verified on
a fully-coupled oscillator sweep through $N=300$, with detailed
single-configuration checks at $N=150$.
(iv)~\texttt{paulikit} is released under GPL-3.0-or-later with
exhaustive per-term checks where claimed, and Work--Span together with
memory-bandwidth ceilings are derived and checked by direct
measurement on a stated low-end GNU/Linux platform
(Section~\ref{sec:results}).

Georges et al.~\citep{georges_et_al_2025} record the cost ladder from
an $O(2^{5n})$ naive trace and an $O(2^{3n})$ naive matmul down to
$O(n\cdot2^{2n})$. Their advance is in-place decomposition on the dense
matrix buffer ($O(1)$ \emph{additional} memory), not a better time
exponent. Among the $O(n\cdot2^{2n})$ peers
\citep{gunlycke_et_al_2020,hantzko_et_al_2024,gidney_2023,
hamaguchi_et_al_2024}, Hamaguchi, Hamada and Yoshioka's \emph{Quantum}
handbook \citep{hamaguchi_et_al_2024} (with open C++ package
\code{paulidecomp} \citep{hamaguchi_et_al_2024_paulidecomp}) gives a
careful early published account of an FWHT-based Pauli-vector routine,
as a subroutine of their Robustness-of-Magic algorithms. Gidney's
practical Stack Exchange construction \citep{gidney_2023}, Hantzko et
al.'s tensorized algorithm \citep{hantzko_et_al_2024}, and Georges et
al.'s dedicated FWHT treatment \citep{georges_et_al_2025} are equally
part of the $O(n\cdot2^{2n})$ lineage.
A separate group Georges et al.\ cite
\citep{romero_santos_suarez_2023,jones_2024,koska_et_al_2024,
lapworth_2022} is named without any complexity attached; reading each
individually shows why not --- none of the four is at $O(n\cdot2^{2n})$
at all.

PauliComposer \citep{romero_santos_suarez_2023} is not really about
decomposition: its target is \emph{building} the matrix of a given
Pauli string, in $O(2^n)$, by exploiting that each row has exactly one
nonzero entry; decomposition rides along as an application. Jones
\citep{jones_2024} takes each of the $4^n$ coefficients head-on, an
$O(2^n)$ Gray-code inner product per coefficient, $O(8^n)$ overall ---
fixed-memory and embarrassingly parallel, and asymptotically worse for
sharing no work between coefficients. Koska et al.\
\citep{koska_et_al_2024} remain $O(4^n)$ but arrange the work as a
recursion tree that prunes repeated sub-products --- a better
constant, not a better exponent. Lapworth~\citep{lapworth_2022} uses
an orthogonality relation for the elementwise (Hadamard) product of
Pauli matrices and its grand sum, aimed at sparsity patterns in
computational-fluid-dynamics matrices; it is not an $O(n\cdot2^{2n})$
transform method.

PennyLane's Pauli-decomposition path~\citep{bergholm_et_al_2022}
moved from a 2020 naive $4^n$-term trace loop for Hermitian operators
only (\path{decompose_hamiltonian}) to a 2023 Walsh--Hadamard rewrite
that drops the Hermitian restriction~\citep{pennylane_pr4395},
following Gidney's construction~\citep{gidney_2023}. The stated
full-matrix cost is $O(n\,4^n)$, the same overall class as the other
FWHT peers. A later sparse path~\citep{pennylane_pr9728} applies the
same dense XOR--Walsh--Hadamard identity as Georges
et al.~\citep{georges_et_al_2025}, restricted to nonempty
$\mathtt{row}\oplus\mathtt{col}$ groups.

Classiq's \texttt{matrix\_to\_hamiltonian}
\citep{classiq_matrix_to_hamiltonian}, absent from Georges et al.'s
own survey, is another production implementation of the same
$\{I,Z\}/\{X,iY\}$-Hadamard-matrix approach, correct and complete from
its introduction, inside a closed-source SDK with no paper or public
documentation to cite in its place. Bisecting PyPI release wheels (the
only way in, absent a public changelog) dates that introduction to
version 0.44.0, uploaded 2024-08-05.

Spencer, Bharti and Gorshkov \citep{spencer_bharti_gorshkov_2026}
also recover a Pauli expansion, but not from an explicit matrix in
hand. Their setting assumes only $k=\mathrm{poly}(n)$ nonzero Pauli
coefficients, and access is by classical sparse queries: each query
returns the nonzero entries of one chosen row or column. A randomized
classical algorithm then reconstructs the support with failure
probability at most $\delta\in(0,1)$, at cost polynomial in $n$ and
$\log(1/\delta)$. By contrast, the $O(n\cdot 4^n)$ matrix-input FWHT methods
\citep{georges_et_al_2025,gidney_2023,hantzko_et_al_2024,hamaguchi_et_al_2024}
make no Pauli-sparsity promise of that kind,
and neither does the coupled-oscillator family of
Section~\ref{sec:motivation}. Those oscillators can be matrix-sparse
--- comparatively few nonzero entries --- yet still Pauli-heavy:
$90{,}431{,}488$ surviving coefficients at $N=150$ ($n=14$) and
$1{,}457{,}121{,}280$ at $N=300$ ($n=16$) --- far from Spencer et
al.'s small-$k$ regime (e.g.\ $k\sim 2n$).

Among the matrix-input FWHT peers, the sharper distinction for
\texttt{paulikit} is memory. Georges et al.'s $O(1)$ \emph{additional}
memory sits on a mandatory dense $2^n\times2^n$ load. On the platform
of Section~\ref{sec:results}, the streaming implementation's measured
peak memory tracks the chunk working set rather than a densified
$2^n\times2^n$ operator, as checked on the oscillator sweep through
$N=300$ (per-term checks and chunking demos at $N=150$). Spencer et
al.\ differ by access model and are outside this matrix-input memory
comparison.

\section{Character-theoretic derivation}
\label{sec:approach}

A Fourier transform diagonalizes \emph{translation by a group
element}. For a finite abelian group $G$, write
$\mathbb{C}^G\cong\mathrm{Fun}(G,\mathbb{C})$
for the vector space of
functions $f:G\to\mathbb{C}$, with $G$ acting on it by translation,
$(g\cdot f)(h)=f(g^{-1}h)$, and let
$\widehat{G}=\{\chi:G\to\mathbb{C}^\times\}$ be its group of
characters. Any operator $T$ commuting with the action satisfies
$T\chi=\lambda_\chi\chi$ for every $\chi\in\widehat{G}$ and some
$\lambda_\chi\in\mathbb{C}$. This means $T$ is diagonal in the
character basis, whatever it looks like in the basis
$\{\delta_g\}_{g\in G}$ it started in. Every transform in this paper
--- the discrete Fourier transform, the Walsh--Hadamard transform ---
is this same diagonalization by characters, applied to a different
group $G$.

Specializing to qubits, take the configuration group
$G=(\mathbb{Z}_2)^n$ and the computational space
$\mathbb{C}^G\cong\mathrm{Fun}(G,\mathbb{C})$ from above.
The characters of $G$ are the sign functions
\begin{equation}
\chi_z(v)=(-1)^{\langle v,z\rangle}, \qquad z\in G,
\label{eq:character}
\end{equation}
with $\langle v,z\rangle=\sum_i v_iz_i\bmod 2$. The characters
$\{\chi_z\}_{z\in G}$ form a basis of $\mathbb{C}^G$.
Each is implemented on the computational space by the multiplication
operator
$Z^{\otimes z}:=Z^{z_1}\otimes\cdots\otimes Z^{z_n}$
(with $Z^0:=I$), acting by
$\bigl(Z^{\otimes z}f\bigr)(v)=\chi_z(v)\,f(v)$.
Thus an operator $T$ diagonal in the computational basis ---
equivalently, supported on the fiber $\Delta_0=\{(v,v):v\in G\}$ with
fiber function $f_0(v)=T_{v,v}$ --- expands uniquely as
\begin{equation}
T=\sum_{z\in G}c_z\,Z^{\otimes z},
\qquad
c_z=2^{-n}\sum_{v\in G}\chi_z(v)\,f_0(v),
\label{eq:diagonal-pauli-expansion}
\end{equation}
where the $c_z$ are exactly the character coefficients of $f_0$, at no
further cost. The real problem begins only where $T$ fails to commute
with translations: nonzero matrix elements off the main diagonal
require shift labels $x$ and the operators $X^{\otimes x}$, taken up
next.

For $x\in G$, the translation operators are
$X^{\otimes x}:=X^{x_1}\otimes\cdots\otimes X^{x_n}$, acting on
$\mathbb{C}^G$ by
$(X^{\otimes x}f)(v)=f(v\oplus x)$, identifying the computational
basis vector $\lvert v\rangle$ with $\delta_v$.
Nonzero matrix elements
$T_{v,v\oplus x}=\bra{v}T\ket{v\oplus x}$ for $x\neq 0$ live on the
fiber $\Delta_x=\{(v,v\oplus x):v\in G\}$, with fiber function
$f_x(v):=T_{v,v\oplus x}\in\mathbb{C}^G$ --- the $x=0$ case
$\Delta_0$, $f_0$ already treated above. The displacements
$X^{\otimes x}Z^{\otimes z}$ combine that shift with multiplication
by $\chi_z$, so Pauli labels are shift--character pairs $(x,z)\in G\times G$.
These labels do not multiply componentwise without a phase: since
$XZ=-ZX$ on one qubit,
\begin{equation}
X^{\otimes x}Z^{\otimes z}
=
(-1)^{\langle x,z\rangle}
Z^{\otimes z}X^{\otimes x}.
\label{eq:symplectic-commutation}
\end{equation}
The factor $(-1)^{\langle x,z\rangle}$ is the value of the symplectic
pairing $(x,z)\mapsto(-1)^{\langle x,z\rangle}$ of shift and character
labels, which encodes the Pauli group's noncommutativity. The bare
products $X^{\otimes x}Z^{\otimes z}$ are real matrices and span only
a real $4^n$-dimensional subspace of $\mathbb{C}^{2^n\times2^n}$.
Admitting $Y:=iXZ$ wherever $x$ and $z$ overlap yields the genuine
strings (as in Proposition~\ref{prop:qubit-decomposition} below)
\begin{equation}
\mathrm{Genuine}(x,z)
=
i^{|x\wedge z|}
X^{\otimes x}Z^{\otimes z},
\qquad
|x\wedge z|
:=
\sum_i x_iz_i\in\{0,\ldots,n\},
\label{eq:genuine-def}
\end{equation}
which do span the full matrix space. (Here $|x\wedge z|$ is an
unreduced integer, unlike the mod-$2$ pairing $\langle x,z\rangle$.)
The Pauli group is then the central extension of $G\times G$ by
$\{\pm1,\pm i\}$ determined by that symplectic pairing --- made
precise in Lemma~\ref{lem:central-extension} below. (The phase
section is not forced by the algebra alone; an unphased real basis is
equally complete. Remark~\ref{rem:why-genuine} returns to
why $Y$ is used here.)

Extracting $T$'s Pauli coefficients from these entries takes two
steps: partition the nonzero matrix elements over the fibers
$\Delta_x$, forming the fiber functions $f_x\in\mathbb{C}^G$ already
introduced, then expand each $f_x$ in the character basis
$\{\chi_z\}_{z\in G}$ and phase by $i^{|x\wedge z|}$. Explicitly,
\begin{equation}
T=\sum_{x,z\in G}d_{x,z}\,\mathrm{Genuine}(x,z),
\qquad
d_{x,z}
=
i^{|x\wedge z|}\,
2^{-n}\sum_{v\in G}\chi_z(v)\,f_x(v),
\label{eq:fiber-decomposition}
\end{equation}
with $\mathrm{Genuine}(x,z)$ as in Eq.~\eqref{eq:genuine-def}. Set the
partition and the phase aside, and each fiber's sum
$2^{-n}\sum_{v\in G}\chi_z(v)\,f_x(v)$ is the ordinary
fast Walsh--Hadamard transform of $f_x$ --- made precise, with its
recursive cost, in the proof of Proposition~\ref{prop:qubit-decomposition}
below. Recovering every Pauli coefficient of $T$ without visiting all
$4^n$ independently --- this paper's problem --- is that transform,
run once per fiber with support in $T$, carrying the symplectic phase of
the Pauli group it decomposes. Before the closed-form statement,
Lemma~\ref{lem:central-extension} makes precise the sense of
\emph{extension} used above; the same notion recurs for qudits as
well as qubits.

\begin{lemma}
\label{lem:central-extension}
Let $Q$ be a finite abelian group and $E$ a group with a central
subgroup~\mbox{$N\leq Z(E)$} (the center of~$E$). If
$\varphi:E\to Q$ is a surjective homomorphism with kernel exactly
$N$, then $N$ is normal in $E$ (it is central), $E/N$ is a group, and
$\varphi$ descends to an isomorphism $E/N\cong Q$: $E$ is a central
extension of $Q$ by $N$.
\end{lemma}

\begin{proof}
Every central subgroup is normal, since $gng^{-1}=n$ for
$n\in N\leq Z(E)$ and every $g\in E$; hence $E/N$ is a group under
coset multiplication. The first isomorphism theorem applied to
$\varphi$ gives $E/\ker\varphi\cong\mathrm{im}\,\varphi$, i.e.\
(since $\ker\varphi=N$ and $\varphi$ is surjective) $E/N\cong Q$.
\end{proof}

The qubit Pauli group instantiates Lemma~\ref{lem:central-extension}:
take $Q=(\mathbb{Z}_2)^n\times(\mathbb{Z}_2)^n$, $N=\{\pm1,\pm i\}$,
and $E=\{\omega\,X^{\otimes x}Z^{\otimes z}:\omega\in N,\,x,z\in(\mathbb{Z}_2)^n\}$,
with $\varphi:E\to Q$ the forget-the-phase map
$\omega\,X^{\otimes x}Z^{\otimes z}\mapsto(x,z)$.
Then $N=Z(E)=\ker\varphi$, so $E/N\cong Q$: Pauli labels $(x,z)$ are
exactly the cosets of $N$ in $E$. The same lemma covers the odd-$p$
Heisenberg--Weyl group in Remark~\ref{rem:parity}; Proposition~
\ref{prop:qubit-decomposition}'s coefficient formula (Steps 1--2
below) does not depend on either instance.

\begin{proposition}[Restatement of Georges et al.
\cite{georges_et_al_2025}, Theorem~1]
\label{prop:qubit-decomposition}
For $x,z\in G=(\mathbb{Z}_2)^n$, write
$\mathrm{Bare}(x,z):=X^{\otimes x}Z^{\otimes z}$ and
$\mathrm{Genuine}(x,z):=i^{|x\wedge z|}\,\mathrm{Bare}(x,z)$
as in Eq.~\eqref{eq:genuine-def} --- the corresponding string in the
genuine Pauli basis, taking $Y:=iXZ$ wherever $x_i=z_i=1$.
Every $T\in\mathbb{C}^{2^n\times2^n}$ has a unique expansion
\begin{equation}
T=\sum_{x,z\in G}d_{x,z}\,\mathrm{Genuine}(x,z).
\label{eq:pauli-decomposition}
\end{equation}
For each fixed $x$, let $f_x\in\mathbb{C}^G$ be the fiber function
$f_x(v):=T_{v,v\oplus x}$ on $\Delta_x$, and let $e_{x,\cdot}$ be its
(normalized) Walsh--Hadamard transform:
\begin{equation}
e_{x,z}=\frac{1}{2^n}\sum_{v\in G}\chi_z(v)\,f_x(v)
=\frac{1}{2^n}\sum_{v\in G}(-1)^{\langle v,z\rangle}T_{v,v\oplus x}.
\label{eq:wht-inversion}
\end{equation}
Then
\begin{equation}
d_{x,z}=i^{|x\wedge z|}e_{x,z}
\label{eq:prop1-conclusion}
\end{equation}
for every $x,z\in G$, recovering Eq.~\eqref{eq:fiber-decomposition}.
Computing every $d_{x,z}$ this way --- $2^n$ independent
length-$2^n$ fast Walsh--Hadamard transforms, one per fiber
function $f_x$ --- costs $O(n\,4^n)$, against $O(2^{5n})$
for evaluating each of the $4^n$ coefficients from its defining trace
independently.
\end{proposition}

\begin{proof}
\emph{Step 1: the character sum, in the bare basis.} Expand $T$ in
the bare basis as
$T=\sum_{x',z'}c_{x',z'}\,\mathrm{Bare}(x',z')$, where
$c_{x',z'}$ is the coefficient of the (possibly non-Hermitian) bare
product $\mathrm{Bare}(x',z')$, as opposed to the proposition's
phased coefficient $d_{x',z'}$. Fix $x$, and let
$f_x(v):=T_{v,v\oplus x}$ be the corresponding fiber function.
Since
$X^{\otimes x}\ket{w}=\ket{w\oplus x}$ for every $w$, the matrix element
$\langle v|\mathrm{Bare}(x,z)|w\rangle$ is nonzero only at $w=v\oplus x$
(on $(\mathbb{Z}_2)^n$, $-x=x$), and the nonzero value is
$\chi_z(w)=\chi_z(v\oplus x)=\chi_z(v)\chi_z(x)$.
On the fiber $\Delta_x$: if $x'\neq x$, then
$\langle v|\mathrm{Bare}(x',z')|v\oplus x\rangle=0$ for every $v$ and
$z'$, so only the summands with offset $x$ contribute to $f_x$, and
they give
\begin{equation}
f_x(v) = \sum_z c_{x,z}\,\chi_z(x)\,\chi_z(v),
\label{eq:group-expansion}
\end{equation}
the character expansion of the fiber function $f_x$.
Orthogonality of $\{\chi_z\}_{z\in G}$ inverts that
expansion to
\begin{equation}
\frac{1}{2^n}\sum_v\chi_z(v)\,f_x(v) = c_{x,z}\,\chi_z(x),
\label{eq:inversion-applied}
\end{equation}
i.e.\ Eq.~\eqref{eq:wht-inversion}'s $e_{x,z}$ satisfies
\begin{equation}
e_{x,z}=c_{x,z}\,\chi_z(x)=c_{x,z}\,(-1)^{\langle x,z\rangle}.
\label{eq:step1-conclusion}
\end{equation}

\emph{Step 2: from bare to genuine.} Every $\mathrm{Bare}(x,z)$ is a
\emph{real} matrix, so $\{\mathrm{Bare}(x,z)\}_{x,z}$ spans only the
real-linear subspace of $\mathbb{C}^{2^n\times2^n}$ --- real dimension
$4^n$, half of $\mathbb{C}^{2^n\times2^n}$'s real dimension $2\cdot
4^n$ --- $T$'s ``real-symmetric part.''
Admitting $Y:=iXZ$ at each overlapping qubit supplies the
missing purely-imaginary complement, at the cost of one factor of $i$
per such qubit --- which is exactly why
Eq.~\eqref{eq:genuine-def} defines
$\mathrm{Genuine}(x,z):=i^{|x\wedge z|}\,\mathrm{Bare}(x,z)$.
Equating the two expansions
$T=\sum_{x,z}c_{x,z}\,\mathrm{Bare}(x,z)=\sum_{x,z}d_{x,z}\,
\mathrm{Genuine}(x,z)$ therefore gives
\begin{equation}
c_{x,z}=i^{|x\wedge z|}d_{x,z}.
\label{eq:bare-to-genuine}
\end{equation}
Combined with Step 1's Eq.~\eqref{eq:step1-conclusion} and the
identity $(-1)^{\langle x,z\rangle}=(-1)^{|x\wedge z|}$, this gives
\begin{equation}
d_{x,z}=e_{x,z}\,(-1)^{\langle x,z\rangle}i^{-|x\wedge z|}
= i^{|x\wedge z|}e_{x,z},
\label{eq:step2-chain}
\end{equation}
the proposition's stated relation, Eq.~\eqref{eq:prop1-conclusion}.

\emph{Step 3: the transform is fast.} For the fixed $x$, the fiber
function $f_x\in\mathbb{C}^G$ is already $f_x(v)=T_{v,v\oplus x}$ ---
one entry of $T$ per $v$ on $\Delta_x$, which
Eq.~\eqref{eq:wht-inversion} sums over. Writing
\begin{equation*}
\hat f_x(z):=\sum_v(-1)^{\langle v,z\rangle}f_x(v)
\end{equation*}
for the ordinary (unnormalized) Walsh--Hadamard transform of $f_x$,
that equation says exactly
\begin{equation}
e_{x,z}=\frac{1}{2^n}\hat f_x(z)
\label{eq:wht-of-f}
\end{equation}
Evaluating $\hat f_x$ from that sum at every
$z\in G$ costs $O(4^n)$. It is fast because it recurses:
writing $v=(v_1,v')$ with $v_1\in\mathbb{Z}_2$ and
$v'\in(\mathbb{Z}_2)^{n-1}$, fixing $v_1$ restricts $f_x$ to two
half-domain functions
\begin{equation*}
f_x^{(0)}(v'):=f_x(0,v'), \qquad f_x^{(1)}(v'):=f_x(1,v'),
\end{equation*}
each on $(\mathbb{Z}_2)^{n-1}$, and write
$\widehat{f}_x^{(0)},\widehat{f}_x^{(1)}$ for their length-$2^{n-1}$
Walsh--Hadamard transforms. Writing $z=(z_1,z')$ the same way as $v$,
$\hat f_x$ is recovered from
$\widehat{f}_x^{(0)},\widehat{f}_x^{(1)}$ by
\begin{equation}
\hat f_x(0,z') = \widehat{f}_x^{(0)}(z')+\widehat{f}_x^{(1)}(z'),
\qquad
\hat f_x(1,z') = \widehat{f}_x^{(0)}(z')-\widehat{f}_x^{(1)}(z'),
\label{eq:wht-butterfly}
\end{equation}
a butterfly combining two length-$2^{n-1}$ transforms into one of
length $2^n$ at the cost of one addition and one subtraction per
entry of $z'$, not a fresh transform. This is the same
halving-and-combining structure as the ordinary FFT, over
$(\mathbb{Z}_2)^n$ instead of $\mathbb{Z}_N$: each of the $n$
recursion levels touches every one of $2^n$ entries once, giving
$O(n2^n)$ for one fixed fiber, and $O(n\,4^n)$ once this is repeated
over all $2^n$ values of $x$.

Compare this to computing every coefficient directly from its
defining projection,
\begin{equation}
d_{x,z} = \mathrm{Tr}\big(T\cdot i^{|x\wedge z|}X^{\otimes x}Z^{\otimes z}\big)\big/2^n,
\label{eq:direct-trace}
\end{equation}
with no transform at all: one
trace per $(x,z)$ pair, $4^n$ pairs in total. Each trace, formed as a
full $2^n\times2^n$ matrix product before summing the diagonal, costs
$O(2^{3n})$, for a total of $O(2^{2n})\cdot O(2^{3n})=O(2^{5n})$ ---
the naive baseline the fast transform above replaces.
\end{proof}

\begin{remark}[Why the genuine section]
\label{rem:why-genuine}
Over $\mathbb{C}$, the bare products $\{\mathrm{Bare}(x,z)\}$ are already
a basis of $\mathbb{C}^{2^n\times2^n}$: Step~1 recovers every
$c_{x,z}$ at the same $O(n\,4^n)$ transform cost, and Step~2's map
$c_{x,z}=i^{|x\wedge z|}d_{x,z}$ is only a phase reparametrization.
(The ``real-linear subspace'' remark in Step~2 concerns real
coefficients and real matrices, not a failure of the $\mathbb{C}$-basis.)
The paper expands in $\{\mathrm{Genuine}(x,z)\}$ because each genuine
term $P$ is Hermitian with $P^2=I$, so $e^{-i\theta P}$ is unitary for
all real $\theta$ and is the elementary Trotter factor for that term.
At an overlapping qubit the bare product $XZ$ is anti-Hermitian,
$(XZ)^\dagger=-XZ$, with $(XZ)^2=-I$, so $e^{-i\theta\,XZ}$ is unitary
only at $\theta=0$ and is not a Trotter factor for nonzero angle.
Termwise Hermiticity is what matches a decomposition term to a real
$\pm1$ Pauli observable; nothing above requires $T$ itself to be
Hermitian.
\end{remark}

Georges et al.\ prove Proposition~\ref{prop:qubit-decomposition} by
direct tensor factorization of the Pauli matrices; to our knowledge,
its derivation from characters of $G=(\mathbb{Z}_2)^n$ above ---
the $2^n$ fibers $\Delta_x=(0,x)+\Delta_0$ (cosets of the diagonal in
$G\times G$), each expanded in $\{\chi_z\}$ via $f_x$, then phased to
the genuine basis --- has not been given before.

Proposition~\ref{prop:qubit-decomposition}'s own construction --- one
independent length-$2^n$ transform per fiber function $f_x$ ---
already shows what Section~\ref{sec:method} depends on: the fibers
may be processed in any grouping, in any order, without recombining
anything afterward, and still recover exactly the same decomposition.

\begin{lemma}[Chunk independence]
\label{lem:chunk-independence}
Partition $G=(\mathbb{Z}_2)^n=X_1\sqcup\cdots\sqcup X_k$. For each $j$,
let $D_j$ be the set of triples $(x,z,d_{x,z})$ with $x\in X_j$ and
$z\in G$, where $d_{x,z}$ is given by
Eq.~\eqref{eq:prop1-conclusion}. Then $\bigcup_j D_j$ is exactly the
full set of coefficients in Eq.~\eqref{eq:pauli-decomposition}.
Moreover, each $D_j$ is determined by $T$ and $X_j$ alone: computing
it neither requires nor affects any $D_i$ for $i\neq j$.
\end{lemma}
\begin{proof}
For each $x$, the fiber function $f_x(v)=T_{v,v\oplus x}$ is read from
$T$ on $\Delta_x$ alone, so $e_{x,\cdot}$ (Eq.~\eqref{eq:wht-inversion})
and $d_{x,z}=i^{|x\wedge z|}e_{x,z}$ depend only on $(T,x)$. Thus for any
chunk $X_j$, the set $D_j$ is determined by $T$ and $X_j$, with no
reference to $X_i$ for $i\neq j$. Because the $X_j$ partition
$G$, every triple $(x,z,d_{x,z})$ lies in exactly one
$D_j$, and $\bigcup_j D_j$ is the full coefficient set of
Eq.~\eqref{eq:pauli-decomposition}.
\end{proof}

Proposition~\ref{prop:qubit-decomposition} and
Lemma~\ref{lem:chunk-independence} together are the chunked streaming
procedure recorded as Algorithm~\ref{alg:streaming-fwht}.
For each active offset $x$, the inner step forms the fiber function
$f_x$, applies the fast Walsh--Hadamard transform, and phases by
$i^{|x\wedge z|}$ to obtain the genuine coefficients $d_{x,z}$
(Eqs.~\eqref{eq:wht-inversion}--\eqref{eq:prop1-conclusion}).
Lemma~\ref{lem:chunk-independence} supplies the outer loop: those
independent per-fiber jobs may be bagged into any partition of the
active offsets, run in any order, and \emph{yielded} as a stream of
triples with no recombination step --- the union across chunks is
still Eq.~\eqref{eq:pauli-decomposition}. An optional threshold
$\tau$ simply drops negligible coefficients at yield time; it does
not change the transform.

\begin{algorithm}[t]
\caption{Chunked streaming Pauli decomposition (qubits)}
\label{alg:streaming-fwht}
\begin{algorithmic}[1]
\Require Matrix $T\in\mathbb{C}^{2^n\times 2^n}$ (dense or sparse
entries); $G=(\mathbb{Z}_2)^n$; optional partition
$X_1\sqcup\cdots\sqcup X_k$ of the active offsets
$\{x:f_x\not\equiv 0\}$ (equivalently,
$\{x:\exists v\text{ with }T_{v,v\oplus x}\neq 0\}$); optional
threshold $\tau\geq 0$ (default $0$).
\Ensure Stream of triples $(x,z,d_{x,z})$ whose union is the
coefficient set of Eq.~\eqref{eq:pauli-decomposition} (entries with
$|d_{x,z}|<\tau$ omitted when $\tau>0$).
\If{no partition is supplied}
\State let $X_1$ be the set of all active offsets and set
$k\leftarrow 1$
\EndIf
\For{$j=1,\ldots,k$}
\ForAll{$x\in X_j$}
\State $f_x\leftarrow\bigl(T_{v,v\oplus x}\bigr)_{v\in G}$
\Comment{fiber function on $\Delta_x$}
\State $e_{x,\cdot}\leftarrow 2^{-n}\,\mathrm{FWHT}(f_x)$
\Comment{Eq.~\eqref{eq:wht-inversion}}
\ForAll{$z\in G$}
\State $d_{x,z}\leftarrow i^{|x\wedge z|}\,e_{x,z}$
\Comment{Eq.~\eqref{eq:prop1-conclusion}}
\If{$|d_{x,z}|\ge\tau$}
\State \textbf{yield} $(x,z,d_{x,z})$
\EndIf
\EndFor
\EndFor
\EndFor
\end{algorithmic}
\end{algorithm}

The same character-theoretic route, run over
$G=(\mathbb{Z}_p)^n$ instead of $(\mathbb{Z}_2)^n$, decomposes an
operator on $n$ qudits of prime dimension $p$ in the analogous
generalized-Pauli basis $X^{\otimes x}Z^{\otimes z}$
\cite{gottesman_1998_higher_dim} --- clock-and-shift matrices of this
general type go back to Sylvester's \emph{nonions}
\cite{sylvester_1882_nonions} --- where
\begin{equation}
X\ket{v} = \ket{v+1\bmod p}, \qquad
Z\ket{v} = \omega^v\ket{v}, \qquad
\omega = e^{2\pi i/p},
\label{eq:qudit-XZ}
\end{equation}
with no restriction to $p=2$ anywhere in the argument. Shift labels
$x\in G$ again partition $T$ into $p^n$ fibers
$\Delta_x=\{(v,v-x):v\in G\}$ with fiber functions
$f_x(v):=T_{v,v-x}\in\mathbb{C}^G$ --- support at $w=v-x$, not
$v+x$, since $-x\neq x$ in general when $p$ is odd.

\begin{corollary}
\label{cor:qudit-decomposition}
Let $p$ be prime, $G=(\mathbb{Z}_p)^n$, $\omega=e^{2\pi i/p}$, and
$\langle v,z\rangle=\sum_i v_iz_i\bmod p$. Every
$T\in\mathbb{C}^{p^n\times p^n}$ has a unique expansion
\begin{equation}
T=\sum_{x,z\in G}d_{x,z}\,X^{\otimes x}Z^{\otimes z}
\label{eq:qudit-decomposition}
\end{equation}
in the Heisenberg--Weyl basis of Eq.~\eqref{eq:qudit-XZ} (no genuine
phase: the un-normalized products already span). For each fixed $x$,
let $f_x\in\mathbb{C}^G$ be the fiber function
$f_x(v):=T_{v,v-x}$ on $\Delta_x=\{(v,v-x):v\in G\}$, and let
$e_{x,\cdot}$ be its discrete Fourier transform over $G$:
\begin{equation}
e_{x,z}=\frac{1}{p^n}\sum_{v\in G}\omega^{-\langle v,z\rangle}f_x(v)
=\frac{1}{p^n}\sum_{v\in G}\omega^{-\langle v,z\rangle}T_{v,v-x}.
\label{eq:qudit-dft-inversion}
\end{equation}
Then
\begin{equation}
d_{x,z}=\omega^{\langle x,z\rangle}e_{x,z}
\label{eq:cor1-conclusion}
\end{equation}
for every $x,z\in G$. Every $d_{x,z}$ is recovered by
$p^n$ independent length-$p^n$ fast Fourier transforms, one per fiber
function $f_x$, in $O(n\,p^{2n}\log p)$ total.
\end{corollary}

\begin{proof}
The same character-sum argument as in
Proposition~\ref{prop:qubit-decomposition}, Step~1, applies for prime
$p$. Since $X^{\otimes x}\ket{w}=\ket{w+x}$, the matrix element
$\langle v|X^{\otimes x}Z^{\otimes z}|w\rangle$ is nonzero only at
$w=v-x$. (At $p=2$, $-x=x$, recovering the $w=v+x$ used there.) That
nonzero matrix element equals
\begin{equation}
\langle v|X^{\otimes x}Z^{\otimes z}|v-x\rangle
=
\omega^{\langle z,v-x\rangle} = \chi_z(v)\,\chi_z(x)^{-1},
\qquad \chi_z(v):=\omega^{\langle v,z\rangle},
\label{eq:qudit-matrix-element}
\end{equation}
with $\{\chi_z\}_{z\in G}$ the characters of $G$.
Now fix $x$ and expand $T$ in the Heisenberg--Weyl basis,
$T=\sum_{x',z'}d_{x',z'}\,X^{\otimes x'}Z^{\otimes z'}$.
Only the summands with shift label $x'=x$ meet the fiber
$\Delta_x$: for $x'\neq x$ one has
$\langle v|X^{\otimes x'}Z^{\otimes z'}|v-x\rangle=0$ for all $v,z'$,
so the fiber function $f_x$ from the corollary statement satisfies
\begin{equation}
f_x(v) = \sum_z d_{x,z}\,\chi_z(x)^{-1}\,\chi_z(v),
\label{eq:qudit-group-expansion}
\end{equation}
the character expansion of $f_x$. Orthogonality of
$\{\chi_z\}_{z\in G}$ inverts that expansion to
\begin{equation}
\frac{1}{p^n}\sum_v\chi_z(v)^{-1}\,f_x(v)
=
d_{x,z}\,\chi_z(x)^{-1},
\label{eq:qudit-inversion-applied}
\end{equation}
i.e.\ Eq.~\eqref{eq:qudit-dft-inversion}'s $e_{x,z}$
(with kernel $\omega^{-\langle v,z\rangle}=\chi_z(v)^{-1}$) satisfies
$e_{x,z}=d_{x,z}\,\chi_z(x)^{-1}$, and therefore
\begin{equation}
d_{x,z} = \chi_z(x)\,e_{x,z} = \omega^{\langle x,z\rangle}e_{x,z}.
\label{eq:qudit-step1-conclusion}
\end{equation}
Step 2 does not arise: the un-normalized basis $\{X^{\otimes
x}Z^{\otimes z}\}$ carries no phase to account for, so
Lemma~\ref{lem:central-extension} has nothing to apply to here --- see
Remark~\ref{rem:parity} below for the odd-$p$ case where it does.
Step 3 is the ordinary radix-$p$ fast
Fourier transform \citep{cooley_tukey_1965} on $G$ in
place of the radix-$2$ Walsh--Hadamard transform, at the standard
cost $O(n\,p^n\log p)$ per fiber, hence $O(n\,p^{2n}\log p)$ over all
$p^n$ values of $x$.

Compare this to computing every coefficient from its Hilbert--Schmidt
projection, with no Fourier transform at all:
\begin{equation}
d_{x,z} = \mathrm{Tr}\big(T\cdot X^{\otimes x}Z^{\otimes z}\big)\big/p^n,
\label{eq:qudit-direct-trace}
\end{equation}
one trace per $(x,z)$ pair, $p^{2n}$ pairs in total. Each trace,
formed as a full $p^n\times p^n$ matrix product before summing the
diagonal, costs $O(p^{3n})$, for a total of
$O(p^{2n})\cdot O(p^{3n})=O(p^{5n})$ --- the same
$\dim^2\cdot\dim^3=\dim^5$ naive scaling as in
Proposition~\ref{prop:qubit-decomposition}, with $\dim=p^n$ in place
of $2^n$.
\end{proof}

The Heisenberg--Weyl products $\{X^{\otimes x}Z^{\otimes z}\}$ used in
Corollary~\ref{cor:qudit-decomposition} (Bare products, when $p=2$)
are already orthogonal in the Hilbert--Schmidt inner product
$\langle A,B\rangle=\mathrm{Tr}(A^\dagger B)$, for every prime $p$.
Orthonormality is only the global rescaling by $p^{-n/2}$.
Phases enter for other reasons --- a packaging section at odd $p$
(Remark~\ref{rem:phase-convention}), and termwise Hermiticity via $Y$
at $p=2$ (Proposition~\ref{prop:qubit-decomposition}).

\begin{remark}
\label{rem:phase-convention}
For odd $p$, the inverse $2^{-1}\bmod p$ exists, so one may replace the
un-normalized generators by the phased section
\begin{equation}
W(a,b):=\omega^{-2^{-1}\langle a,b\rangle}Z^{\otimes a}X^{\otimes b}
\label{eq:W-phase-convention}
\end{equation}
(note $Z$ before $X$; cf.\ \citep{sharma_ghosh_sazim_2024_hw_bloch}).
Expanding $T$ in $\{W(a,b)\}$ yields a recovery formula of the same
single-phase shape as Eq.~\eqref{eq:cor1-conclusion}, with no extra
symplectic sign to track. That is a packaging convention for the
coefficients, not a defect the un-normalized basis had.

The same phase does \emph{not} make individual $W(a,b)$ Hermitian.
For any phase convention,
\begin{equation}
(X^{\otimes a}Z^{\otimes b})^\dagger
=
\omega^{\langle a,b\rangle}
X^{\otimes(-a)}Z^{\otimes(-b)},
\label{eq:dagger-identity}
\end{equation}
so a single term is self-adjoint only if
$X^{\otimes(-a)}Z^{\otimes(-b)}=X^{\otimes a}Z^{\otimes b}$, i.e.\
$2a\equiv 2b\equiv 0\bmod p$. For odd $p$, where $2$ is invertible,
this forces $a=b=0$. At $p=2$, by contrast, $-a\equiv a$ holds for
every $a$, so $Y:=iXZ$ can restore Hermiticity termwise --- as in
Proposition~\ref{prop:qubit-decomposition}.
\end{remark}

\begin{remark}
\label{rem:parity}
Failure of termwise Hermiticity does not block a group law. Let $E$ be
the group generated by $\{W(a,b)\}$ together with the $p$-th roots of
unity, let $Q=(\mathbb{Z}_p)^n\times(\mathbb{Z}_p)^n$, and let
$N=\{zI:z^p=1\}$. The map $\varphi:E\to Q$, $W(a,b)\mapsto(a,b)$, is
a surjective homomorphism with kernel $N$. 
Lemma~\ref{lem:central-extension} therefore applies: $E$ is a central
extension of $Q$ by $N$ --- the same structural statement as for
qubits, with $|N|=p$ in place of $\{\pm1,\pm i\}$'s order~$4$. The
section $s_p(a,b):=W(a,b)$ is the odd-$p$ analogue of the genuine
section $\mathrm{Genuine}(x,z)$ used for qubits in
Proposition~\ref{prop:qubit-decomposition}.
\end{remark}

\section{Case study: coupled harmonic oscillators}
\label{sec:motivation}

As a concrete family on which to measure Pauli-decomposition cost ---
the $O(n\,4^n)$ character-sum / Walsh--Hadamard scale against any
cheaper structure one might have hoped the encoding would retain ---
we take a fully coupled network of $N$ classical oscillators linked by
Hooke's-law springs. Following Babbush et
al.\ \citep{babbush_2023_oscillators}, the harmonic normal-mode problem
$\mu\ddot{x}=-Kx$ with spring-constant matrix $K$ and diagonal mass
matrix $\mu$ is simulated by the time evolution generated by the block
Hamiltonian
\begin{equation}
H:=-\begin{pmatrix}0&B\\B^\dagger&0\end{pmatrix},\qquad
BB^\dagger=\mu^{-1/2}K\mu^{-1/2},
\label{eq:oscillator-hamiltonian}
\end{equation}
on $\mathbb{C}^{N+M}$, with $M$ the number of springs. For the
fully-coupled family used below, each mass has an on-site spring and a
spring to every other mass, so
$M=\binom{N}{2}+N=\binom{N+1}{2}$. (A path graph, with only
$M=N-1$ springs, yields a much smaller $N+M$ and padded qubit count
$n$; it is not the family measured below.) The Hamiltonian dimension
$N+M$ therefore grows quadratically in $N$ before padding to
$n=\lceil\log_2(N+M)\rceil$ qubits:
\begin{table}[h]
\centering
\begin{tabular}{rrrr}
\toprule
$N$ (oscillators) & $M=\binom{N+1}2$ (springs) & $N+M$ (Hamiltonian dim.) & $n$ (qubits) \\
\midrule
2  & 3    & 5    & 3 \\
4  & 10   & 14   & 4 \\
8  & 36   & 44   & 6 \\
16 & 136  & 152  & 8 \\
32 & 528  & 560  & 10 \\
\bottomrule
\end{tabular}
\caption{Oscillator count against Hamiltonian dimension and qubit
count, for the fully-coupled family.}
\label{tab:oscillator-sizing}
\end{table}

Because $N+M=\Theta(N^2)$, one has $n=O(\log N)$: doubling $N$
increases the qubit count by about two
(Table~\ref{tab:oscillator-sizing}) --- each such step multiplies the
$O(n\,4^n)$ decomposition cost by sixteen. A network with tens of
masses already reaches the $n\ge10$ regime (e.g.\ $N=32$ gives $n=10$
in the table). At those sizes every entry of $H$ is still fixed by the
network's physical parameters alone, so extracting its Pauli
decomposition is a pure linear-algebra step with no further physics in
it. The family is therefore a convenient benchmark for
Pauli-decomposition cost.

\section{Memory-bounded streaming implementation}
\label{sec:method}

Proposition~\ref{prop:qubit-decomposition} establishes a time bound ---
$O(n\cdot2^n)$ per transform and $O(n\cdot4^n)$ overall --- and nothing
beyond that: neither peak memory, nor the leading constant $C$ in the
$O(n\cdot4^n)$ estimate (runtime at most $C\,n\,4^n$ for large $n$),
nor multi-core behaviour. The present section makes those three
quantities concrete for the deployed streaming realization of
Algorithm~\ref{alg:streaming-fwht}, and records how the implementation
is checked against the proposition.

Fix a size convention: each surviving coefficient is stored as an
$(x,z,\text{coefficient})$ triple --- two 64-bit indices and one
128-bit complex value --- hence 32 bytes. On the oscillator family of
Section~\ref{sec:motivation} at $N=150$ ($n=14$ qubits,
$\dim=16384$), the operator has $90{,}431{,}488$ nonzero Pauli
coefficients at the library default $\texttt{atol}=10^{-10}$
(\texttt{paulikit}~$0.1.4$), so at $32$~bytes each those triples
alone cost $\approx2.7$~GiB. A naive
implementation that materializes every fiber function
$f_x(v)=T_{v,v\oplus x}$ before any Walsh--Hadamard transform
(Eq.~\eqref{eq:wht-inversion}) would allocate a second buffer of
comparable size. At this $N$ the sparse matrix has $45{,}000$
nonzeros on $n_{\text{active}}:=11{,}039$ distinct offsets~$x$;
stacking the corresponding length-$\dim$ fibers as one dense
$(n_{\text{active}},\dim)$ complex128 array costs
$n_{\text{active}}\times\dim\times16\text{ bytes}=11{,}039\times
16384\times16\text{ bytes}\approx2.7$~GiB on top of the answer.
At this size both footprints already reach multi-gigabyte workstation
scale. The subsections below close the gap on memory, leading
constants, and parallelism; correctness is checked separately.

\subsection{Peak memory and streaming accumulation}

Under Algorithm~\ref{alg:streaming-fwht}, peak memory scales as
$O(\mathrm{chunk\_size}\cdot\dim+n_{\text{terms}})$. The first term is
the live working set: at most $\mathrm{chunk\_size}$ fiber functions of
length $\dim$ are built and Walsh--Hadamard-transformed together, then
discarded; in particular, no $(n_{\text{active}},\dim)$ accumulator of
all active fibers is retained. The second term is the retained answer:
if all surviving $(x,z,d)$ triples are kept until return, peak memory
also grows with $n_{\text{terms}}$. Using the fast butterfly~(Eq.~\eqref{eq:wht-butterfly}) does not by
itself reduce peak memory: if every active fiber $f_x$ is still written
into one $(n_{\text{active}},\dim)$ result buffer, that buffer alone
costs $O(n_{\text{active}}\cdot\dim)$, even when most entries are later
thresholded to zero.
Chunking the \emph{input} does not remove that cost if the
\emph{output} is still assembled in full before return.

The memory-bounded path never builds string labels or a
\texttt{dict[str,complex]}; the streaming interface yields raw arrays
of indices $(x,z)$ and coefficients. A dict-returning interface remains
available for callers who want string keys; it is not required by the
design above. Both choices
follow Lemma~\ref{lem:chunk-independence} and
Algorithm~\ref{alg:streaming-fwht}: each chunk is an independent
result, with no recombination step, and the union across chunks equals
the unchunked decomposition (Figure~\ref{fig:chunking}).

\begin{figure}[htbp]
\centering
\resizebox{\linewidth}{!}{

\begin{tikzpicture}[
    font=\small,
    node distance=8mm and 10mm,
    chunkbox/.style={
      draw, rounded corners, minimum width=32mm, minimum height=8mm,
      align=center
    },
    workbox/.style={
      draw, rounded corners, minimum width=32mm, minimum height=22mm,
      text width=31mm, align=left, fill=black!4, font=\scriptsize,
      inner sep=2mm
    },
    arr/.style={-{Stealth[length=2mm]}, thick}
  ]

  \node[draw, minimum width=132mm, minimum height=10mm, align=center]
    (xset) {active-row index set $X\subseteq(\mathbb{Z}_2)^n$, $|X|=n_{\text{active}}$};

  \node[chunkbox, below=of xset, xshift=-50mm] (x1) {$X_1$};
  \node[chunkbox, right=of x1] (x2) {$X_2$};
  \node[right=of x2, font=\large] (dots) {$\cdots$};
  \node[chunkbox, right=of dots] (x4) {$X_C$};

  \draw[arr] (xset.south) -| (x1.north);
  \draw[arr] (xset.south) -| (x2.north);
  \draw[arr] (xset.south) -| (x4.north);

  \node[align=right, left=4mm of x1, font=\scriptsize, text width=20mm]
    {chunk\_size\\$\approx|X_1|$};

  \newcommand{\workstages}[1]{%
    \textbf{1. gather} rows $T_{v,v+x}$, $x\in#1$\\
    \textbf{2. butterfly} (Eq.~\ref{eq:wht-butterfly})\\
    \textbf{3. phase} multiply by $i^{|x\wedge z|}$\\
    \textbf{4. threshold} drop $|d_{x,z}|<\text{tol}$%
  }
  \node[workbox, below=10mm of x1] (w1) {\workstages{X_1}};
  \node[workbox, below=10mm of x2] (w2) {\workstages{X_2}};
  \node[workbox, below=10mm of x4] (w4) {\workstages{X_C}};

  \draw[arr] (x1) -- (w1);
  \draw[arr] (x2) -- (w2);
  \draw[arr] (x4) -- (w4);

  \node[left=3mm of w1, align=right, font=\scriptsize, text width=16mm]
    {independent (Lemma~\ref{lem:chunk-independence})};

  \node[chunkbox, below=10mm of w1, fill=black!8]
    (d1) {$(x,z,\text{coeff})_{X_1}$};
  \node[chunkbox, below=10mm of w2, fill=black!8]
    (d2) {$(x,z,\text{coeff})_{X_2}$};
  \node[chunkbox, below=10mm of w4, fill=black!8]
    (d4) {$(x,z,\text{coeff})_{X_C}$};

  \draw[arr] (w1) -- (d1);
  \draw[arr] (w2) -- (d2);
  \draw[arr] (w4) -- (d4);

  \node[draw, minimum width=132mm, minimum height=10mm, align=center,
        below=10mm of d1, xshift=50mm, fill=black!12]
    (union) {$D_1\cup D_2\cup\cdots\cup D_C$ = full decomposition
             (Lemma~\ref{lem:chunk-independence})};

  \draw[arr] (d1.south) -- ++(0,-5mm) -| (union.north -| d1);
  \draw[arr] (d2.south) -- ++(0,-5mm) -| (union.north -| d2);
  \draw[arr] (d4.south) -- ++(0,-5mm) -| (union.north -| d4);

\end{tikzpicture}}
\caption{Chunking as partition-and-stream: the active-offset index set
$X$ is split into disjoint chunks $X_1,\ldots,X_C$
(Lemma~\ref{lem:chunk-independence}), each processed independently
into raw $(x,z,\text{coefficient})$ arrays, with no recombination step
--- the union across chunks is exactly the unchunked decomposition.
Schematic: three representative chunks are drawn, not the true chunk
count $C$ for any specific problem size.}
\label{fig:chunking}
\end{figure}
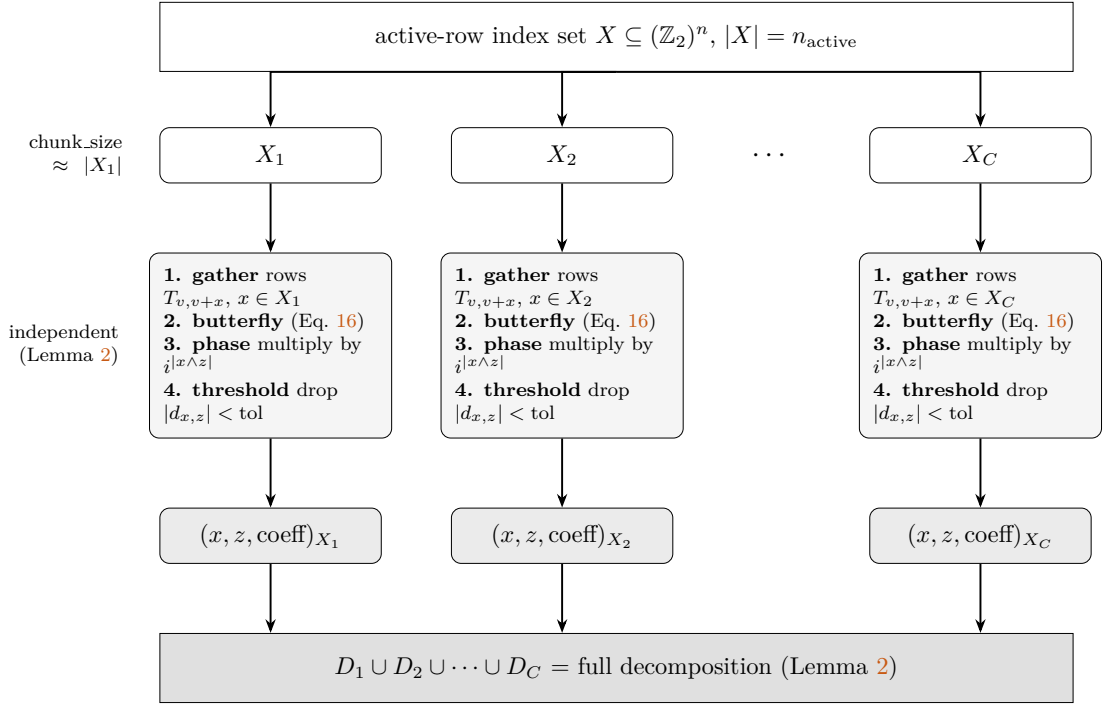

Measured peak memory under this design --- including the $N=150$
case above and the oscillator sweep through $N=300$ --- is reported
with protocol in Section~\ref{sec:results}.

\subsection{Leading constants and kernel cost}
\label{sec:leading-constants-kernels}

A pure-Python profile of this workload puts the leading runtime cost
not in the Walsh--Hadamard butterfly counted in
Proposition~\ref{prop:qubit-decomposition}, but in per-survivor label
construction (mapping each numeric $(x,z)$ to an ``IXYZ''-style
string). That step lies outside the proposition's $O(n\,4^n)$
transform count. The remedy is compiled kernels for labeling and for
the per-fiber transform --- the same coefficients, with a smaller
leading constant. Throughput under those kernels, and which parallel
mechanism wins at which scale, is measured in
Section~\ref{sec:results}; implementation correctness is checked in
Section~\ref{sec:correctness-verification}.

\subsection{Parallel execution}

Chunk independence (Lemma~\ref{lem:chunk-independence}) is what makes
streaming memory-bounded: each chunk is yielded alone, with no shared
fiber accumulator. The same independence lets chunks run concurrently
on multiple cores. Under CPython the two usual mechanisms for
CPU-bound chunk work are child processes and threads; on the
array-yielding path \texttt{paulikit} provides both --- a process pool
and a threaded drain.

Threads are always available in CPython, but under GIL-bound chunk
work they do not deliver multi-core concurrency: the lock serializes
Python bytecode. For that workload, multi-core concurrency requires
child processes. A \emph{process pool} runs each chunk in its own
process and returns the result arrays to the parent; chunk
independence is still Lemma~\ref{lem:chunk-independence}. The pool
pays only while per-chunk compute exceeds the inter-process round
trip that ships inputs and results. On the memory side, each worker
duplicates a working set, so peak RSS can scale with the pool size
even though the logical stream is unchanged.

Once the compiled Walsh--Hadamard and coefficient kernels both
release the GIL and dominate runtime, a threaded drain can overlap
chunk transforms in a single address space, with no large result
arrays crossing a process boundary. Peak RSS then stays one process
footprint while chunks still run concurrently, with the same
yield-per-chunk interface as the sequential stream. Without those
kernels, threads cannot overlap useful work, so the dict-returning
interface (pure-Python labeling on the hot path) stays process-based.
On the parallel array-yielding path, \texttt{paulikit} defaults to a
thread drain when \emph{both} of those kernels are present in the
build, and to a process pool if either is missing; which mechanism is
faster at which scale is measured in Section~\ref{sec:results}.

\subsection{Correctness and verification}
\label{sec:correctness-verification}

To check that the streaming implementation matches
Proposition~\ref{prop:qubit-decomposition}, every coefficient that
survives the implementation's threshold is compared to the defining
Hilbert--Schmidt projection~\eqref{eq:direct-trace} on a code path
that does not use the Walsh--Hadamard transform --- not a sample of
survivors, and not the discarded near-zeros. The comparison runs on
configurations up to the full $N=150$ case above (the same
single-point check cited for chunking demos). Over those survivors
the maximum absolute deviation
$\max\lvert c_{\mathrm{impl}}-c_{\mathrm{ref}}\rvert$ is
$1.04\times10^{-17}$ in \texttt{complex128} arithmetic. Memory and
timing figures that quote this $N=150$ configuration are reported only
for outputs that pass the check. The $N=300$ sweep measures memory and
timing only; it is not a second exhaustive per-term audit.

\section{Results}
\label{sec:results}

The streaming implementation of Section~\ref{sec:method} is evaluated
against two derived ceilings and against direct measurement. The first
ceiling is the concurrency admitted by the dependency graph
(Work--Span); the second is the single-core limit set by the memory
hierarchy (arithmetic intensity versus sustained bandwidth). Both are
read from the structure of the shipped code; hardware parameters enter
only when a ceiling is evaluated numerically on the test platform and
compared with measurement. The subsections that follow report memory
scaling with result size, the process-versus-thread execution model in
practice, and parallel efficiency relative to a single-core baseline.
Where the platform's frequency-scaling governor would dominate
wall-clock comparisons, speedups are quoted in retired-cycle counts
rather than wall time (Sections~\ref{sec:quantifying-work-and-span}
and~\ref{sec:parallel-scaling}).

\subsection{Work--Span bound on parallelism}
\label{sec:theoretical-ceiling}

A Work--Span analysis bounds the concurrency admitted by the
decomposition's dependency structure, prior to any question of how
well a given scheduler exploits it. The bound is usable only when the
two kinds of work in the pipeline are costed as distinct units; that
distinction is fixed next, after which Work, Span, and the numerical
ceiling are obtained.

\subsubsection{Dependency structure}
\label{sec:dependency-structure}

Chunked decomposition (Section~\ref{sec:method}) has a fan-out /
fan-in shape. Chunks are mutually independent
(Lemma~\ref{lem:chunk-independence}): each worker runs gather,
butterfly, phase, and threshold for its own chunk, and every result
converges on a single consumer. How work is assigned differs by
mechanism (Section~\ref{sec:process-vs-thread}): a process pool refills
a bounded in-flight window as results arrive; the threaded drain
gives each worker a static, capacity-weighted share up front, with
work-stealing after that share is exhausted. In both cases completion
order is by arrival, not submission.

The fan-in is where \emph{array-native} work hands off to
\emph{interpreter-native} work. Call the former tier~V and the latter
tier~P (Figure~\ref{fig:dag}). Tier~V is NumPy or the compiled kernels
\texttt{wht\_rows\_inplace} and \texttt{coeffs\_from\_transformed}.
The WHT kernel touches every element of every row at every butterfly
stage, so its unit is one compiled-loop element. The coeffs kernel
matches that on its read side (phase, scale, threshold) but emits one
$(x,z,\text{coefficient})$ triple per survivor on the write side ---
a second, data-dependent unit sized by how many terms clear the
threshold, not by \texttt{chunk\_size} alone. Tier~P is CPython: one
unit is one bytecode-level object operation (hash, refcount, dict
probe). On the CPython side, unpickling is process-boundary-specific
(threads skip it by sharing the parent's address space); per-term
\texttt{str} keys and dict entries are not --- label-hungry callers
pay them on either executor.

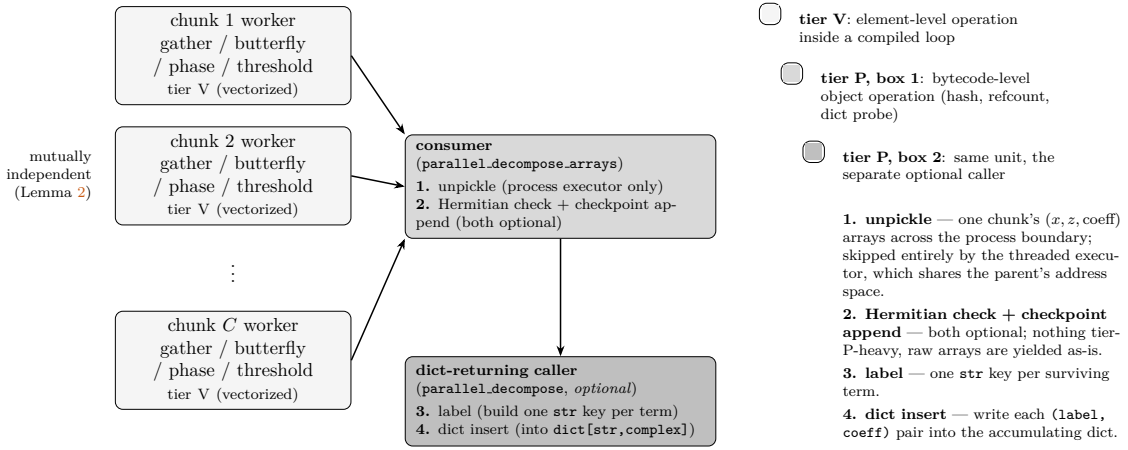
\begin{figure}[htbp]
\centering
\resizebox{\linewidth}{!}{

\begin{tikzpicture}[
    font=\small,
    node distance=4mm and 10mm,
    vnode/.style={
      draw, rounded corners, minimum width=44mm, minimum height=10mm,
      text width=40mm, align=center, fill=black!4
    },
    pnode/.style={
      draw, rounded corners, minimum width=58mm,
      text width=54mm, align=left, fill=black!15, font=\scriptsize,
      inner sep=1.5mm
    },
    arr/.style={-{Stealth[length=2mm]}, thick}
  ]

  \node[vnode] (v1)
    {chunk $1$ worker\\ gather / butterfly / phase / threshold\\
     \scriptsize tier V (vectorized)};
  \node[vnode, below=of v1] (v2)
    {chunk $2$ worker\\ gather / butterfly / phase / threshold\\
     \scriptsize tier V (vectorized)};
  \node[below=4mm of v2] (dots) {$\vdots$};
  \node[vnode, below=4mm of dots] (v4)
    {chunk $C$ worker\\ gather / butterfly / phase / threshold\\
     \scriptsize tier V (vectorized)};

  \node[left=3mm of v2, align=right, font=\scriptsize, text width=16mm]
    {mutually independent (Lemma~\ref{lem:chunk-independence})};

  \node[pnode, right=of v2, yshift=-2mm, minimum height=16mm] (drain)
    {\textbf{consumer} (\texttt{parallel\_decompose\_arrays})\\[1mm]
     \textbf{1.} unpickle (process executor only)\\
     \textbf{2.} Hermitian check + checkpoint append (both optional)};

  \draw[arr] (v1.east) -- (drain.north west);
  \draw[arr] (v2.east) -- (drain.west);
  \draw[arr] (v4.east) -- (drain.south west);

  \node[pnode, below=22mm of drain, fill=black!25,
        minimum height=16mm] (labelcaller)
    {\textbf{dict-returning caller} (\texttt{parallel\_decompose},
     \emph{optional})\\[1mm]
     \textbf{3.} label (build one \texttt{str} key per term)\\
     \textbf{4.} dict insert (into \texttt{dict[str,complex]})};

  \draw[arr] (drain.south) -- (labelcaller.north);

  \node[right=8mm of drain, yshift=33mm] (keyorigin) {};
  \node[draw, rounded corners, minimum width=4mm, minimum height=4mm,
        anchor=north west] (swatchV) at (keyorigin.north west) {};
  \node[fill=black!4, minimum width=3mm, minimum height=3mm] at (swatchV) {};
  \node[right=2mm of swatchV, font=\scriptsize, align=left,
        text width=44mm, anchor=north west, yshift=1.5mm] (labelV)
    {\textbf{tier V}: element-level operation inside a compiled loop};
  \node[draw, rounded corners, minimum width=4mm, minimum height=4mm,
        anchor=north west, below=2mm of labelV.south west] (swatchP1) {};
  \node[fill=black!15, minimum width=3mm, minimum height=3mm] at (swatchP1) {};
  \node[right=2mm of swatchP1, font=\scriptsize, align=left,
        text width=44mm, anchor=north west, yshift=1.5mm] (labelP1)
    {\textbf{tier P, box 1}: bytecode-level object operation (hash,
     refcount, dict probe)};
  \node[draw, rounded corners, minimum width=4mm, minimum height=4mm,
        anchor=north west, below=2mm of labelP1.south west] (swatchP2) {};
  \node[fill=black!25, minimum width=3mm, minimum height=3mm] at (swatchP2) {};
  \node[right=2mm of swatchP2, font=\scriptsize, align=left,
        text width=44mm, anchor=north west, yshift=1.5mm] (labelP2)
    {\textbf{tier P, box 2}: same unit, the separate optional caller};

  \node[below=3mm of labelP2.south west, anchor=north west,
        font=\scriptsize, align=left, text width=52mm]
    (key) {%
     \textbf{1. unpickle} --- one chunk's $(x,z,\text{coeff})$ arrays
     across the process boundary; skipped entirely by the threaded
     executor, which shares the parent's address space.\\[1mm]
     \textbf{2. Hermitian check + checkpoint append} --- both
     optional; nothing tier-P-heavy, raw arrays are yielded as-is.\\[1mm]
     \textbf{3. label} --- one \texttt{str} key per surviving term.\\[1mm]
     \textbf{4. dict insert} --- write each \texttt{(label, coeff)}
     pair into the accumulating dict.%
  };

\end{tikzpicture}}
\caption{Dependency structure underlying the Work--Span analysis: per-chunk
worker nodes (tier V) are mutually independent
(Lemma~\ref{lem:chunk-independence}) and run concurrently, but every
chunk's result passes through a single, totally-ordered consumer
(tier P) in one process, regardless of which drain mechanism
(\S\ref{sec:process-vs-thread}) feeds it. That consumer is minimal for
\texttt{parallel\_decompose\_arrays} itself; the label- and
dict-building cost belongs to a separate, optional caller one layer
up that the array-yielding results reported here skip. This
asymmetry --- many independent V-tier nodes converging on one serial
P-tier consumer --- is what \S\ref{sec:theoretical-ceiling} bounds.}
\label{fig:dag}
\end{figure}

A step's tier is fixed by what it does, not by which API it sits
behind. Wrong tier assignment mis-weights Work and Span: costing a
compiled store as a CPython dict insertion (or the reverse) moves the
apparent critical path and does not cancel between the two sides of
the ratio. Labeling is the sharp example on the dict-returning path:
the wrapper can look compiled yet still build one \texttt{str} per
term, placing a tier~P cost inside tier~V-looking code. (The
array-yielding path measured below retires that cost entirely.)

\subsubsection{Corrected cost model}
\label{sec:corrected-model}

Serial tier~P work puts the critical path at the fan-in. The $C$
tier~V chunks may run concurrently, but every result still converges
on one consumer in one process, for either drain mechanism
(Section~\ref{sec:process-vs-thread}). That consumer is either the
minimal array path in \texttt{parallel\_decompose\_arrays} (optional
Hermitian check and checkpoint) or a heavier dict-returning wrapper
that builds per-term labels and a \texttt{dict[str,complex]} serially
(Section~\ref{sec:method}). Measurements here use the array path. In
both cases Span is one critical-path tier~V chunk plus the consumer's
full tier~P total --- a single totally-ordered node after convergence,
so idle workers cannot steal it and more workers cannot split it.
Reducing that cost means shrinking what the consumer does, as the
array path does by dropping per-term labels and dict construction,
not by adding workers.

\subsubsection{Work and Span}
\label{sec:quantifying-work-and-span}

On the compiled, array-yielding path measured here, Work is
the sum of four concrete costs. Active fiber functions are stored as
rows of a dense gather buffer and split into contiguous chunks of
length at most \texttt{chunk\_size}. The chunks cover
$[0,n_{\text{active}})$ without overlap or gap, so
$\sum_c\mathrm{rows}_c=n_{\text{active}}$. A kernel cost that scales
as $\mathrm{rows}_c\times\dim$ on each chunk therefore totals as one
pass over all $n_{\text{active}}$ fiber functions:
\begin{equation}
\text{Work} \approx t_{\text{wht}}\, n_{\text{active}}\,\dim\log_2\dim
  \;+\; t_{\text{coeffs}}\, n_{\text{active}}\,\dim
  \;+\; t_{\text{term}}\, n_{\text{terms}}
  \;+\; t_{\text{P}}\, C.
\label{eq:total-work-final}
\end{equation}
Here $t_{\text{wht}}$ is the compiled WHT cost per element per
butterfly stage, $t_{\text{coeffs}}$ the compiled coeffs cost per
element in the fused pass, and $t_{\text{P}}$ the per-chunk consumer
cost on the fastest path (dequeue, unpack, yield; no checkpoint or
Hermitian check). The trim rate $t_{\text{term}}$ applies only when
the wrapper actually copies survivors; it is not paid on every chunk. The WHT term is $\Theta(n_{\text{active}}\,\dim\log_2\dim)$. Its
operation count does not depend on the matrix entries: both the NumPy
fallback and the compiled kernel (\texttt{paulikit\_wht\_rows}) run
$\log_2\dim$ doubling stages over each chunk's $(\mathrm{rows},\dim)$
block with no value-dependent branches. The compiled kernel's
\texttt{tile} comes from \texttt{sysconf} --- machine-dependent like
chunk-size auto-tune, not value-dependent. The coeffs term on the compiled path is
$\Theta(n_{\text{active}}\,\dim)$: a single branchless pass
(\texttt{paulikit\_coeffs\_from\_transformed}, the default whenever
extensions are present). Survivor output is a separate
$\Theta(n_{\text{terms}})$ trim in the wrapper, paid only when that
copy triggers. A NumPy fallback sits in the same leading asymptotic
class but is not what the ceiling uses. Tier~P contributes
$\Theta(C)$ only.

At the test platform's $1.8$\,GHz base clock, the compiled rates are
$t_{\text{wht}}=0.749$\,ns and $t_{\text{coeffs}}=5.189$\,ns per
element, and $t_{\text{P}}=234.78$\,ns per chunk, from
\texttt{perf\_event\_open} instruction and cycle counts. At
$\dim=2^{14}$, the WHT term still leads coeffs
($t_{\text{wht}}\log_2\dim\approx10.5$\,ns ${}>\,t_{\text{coeffs}}$
per element); both dominate trim and tier~P on the pinned
configuration used here.

These rates are not taken from wall-clock time. On microsecond
bursts wall-clock was bimodal (spread up to $3.5\times$) under
\texttt{intel\_pstate} frequency stepping, while instruction counts
stayed stable to four figures (spread $<1.001\times$). Cycles still
vary (spread up to $1.69\times$ on the WHT stage in this deposit), so
seconds are recovered with each kernel's measured instructions per
cycle ($2.11$ for WHT, $3.62$ for coeffs). The $t_{\text{P}}$ figure
matches a direct wall-clock check near $\sim$200\,ns across
1--100\,MB payloads (spread $<1.01\times$ on the cycle-derived value).

The trim copy runs only when a chunk's survivor count leaves the
output under half full. At $N=150$ and \texttt{chunk\_size}$=2$, that
happened for $2$ of $5{,}520$ chunks ($0.036\%$); one such chunk
($15{,}872$ survivors) measured $t_{\text{term}}=36{,}463$\,ns at
base clock (instructions per cycle $0.20$, consistent with a
bandwidth-bound $\sim$496\,KiB copy). At this configuration the trim's
contribution to Work is negligible next to the kernel terms --- not
$t_{\text{term}}\times n_{\text{terms}}$ as if every survivor paid a
full trim.

The chunk count is
$C=\lceil n_{\text{active}}/\texttt{chunk\_size}\rceil$.
The default \texttt{chunk\_size} (\texttt{None}) invokes an auto-tune
driven by host cache probes and live \texttt{MemAvailable}; that can
change $C$ between runs, so it is not a stable basis for the ceilings
that follow. All Work/Span figures here pin the chunk size to $2$,
giving $C=5520$ at the pinned $n_{\text{active}}=11{,}039$ used here,
confirmed against the producing script.

Span is one critical-path tier~V chunk plus the full serial tier~P
total (Figure~\ref{fig:dag}). The critical-path chunk is charged at
length \texttt{chunk\_size} --- the largest chunk; the last may be
shorter --- while every chunk still passes through the totally-ordered
consumer, so $t_{\text{P}}\,C$ does not parallelize. The trim term is
kept only as a conservative upper bound (all $n_{\text{terms}}$
survivors charged on that path):
\begin{equation}
\text{Span} \approx t_{\text{wht}}\,\text{chunk\_size}\,\dim\log_2\dim
  \;+\; t_{\text{coeffs}}\,\text{chunk\_size}\,\dim
  \;+\; t_{\text{term}}\, n_{\text{terms}}
  \;+\; t_{\text{P}}\, C.
\label{eq:total-span-final}
\end{equation}
The numerical Span quoted below omits the trim upper bound (negligible
at this configuration; see above) and uses only the critical-path
tier~V chunk plus $t_{\text{P}}\,C$.

\subsubsection{Parallelism ceiling and its limits}
\label{sec:scope-and-limits}

At a finite worker count $p$, the usable ceiling is Brent's bound
\citep{brent_1974},
\begin{equation}
T_p \;\le\; \frac{\text{Work}}{p} + \text{Span},
\label{eq:brent-bound}
\end{equation}
not the unbounded-processor ratio
$\bar{P}=\text{Work}/\text{Span}$ from
Eqs.~\eqref{eq:total-work-final}--\eqref{eq:total-span-final}.
\texttt{paulikit} caps \texttt{n\_workers} at physical core count by
default, so $p$ is small next to $\bar{P}$.

Two caveats. First, classical work-stealing time bounds
\citep{blumofe_leiserson_1999} assume a fully strict spawn/join tree;
the threaded drain is a flat pool with static shares and stealing
between worker queues, so that theorem is not claimed for this
scheduler. Second, Work and Span need the operator's
$n_{\text{active}}$ and $n_{\text{terms}}$, not $\dim$ alone, and the
bound covers \texttt{parallel\_decompose\_arrays} only through its
own yield --- not caller-side work on each chunk.

On the pinned $N=150$ configuration, the cost-model totals from
Eqs.~\eqref{eq:total-work-final}--\eqref{eq:total-span-final} (using
the instruction-derived rates at the $1.8$\,GHz base clock above --- not
an end-to-end wall-clock timer) are $\text{Work}\approx2.836$\,s and
$\text{Span}\approx1.810$\,ms, so $\bar{P}\approx1567$ --- far above
the machine's core count. That Work figure is deliberately a
base-clock reconstruction of the kernel instruction budget; it is
\emph{not} a stopwatch prediction. The same configuration's serial
\texttt{parallel\_decompose\_arrays} wall (one worker,
\texttt{chunk\_size}$=2$, $\texttt{atol}=10^{-10}$, five warmed
repetitions) is a median $1.53$\,s (range $1.43$--$1.64$\,s), so
$\text{Work}/\text{wall}\approx1.85\times$ --- the same order as the
$\sim\!2.2\times$ gap between base-clock and wall microbench rates under
\texttt{intel\_pstate} turbo on this host. Brent at $p=2$ and $p=4$
then bounds speedup by $1.997\times$ and $3.990\times$, nearly ideal,
because Span is dominated by serial tier~P ($71.6\%$ of Span; Work
splits $67\%$/$33\%$ WHT/coeffs).

Speedups below are retired-cycle counts (the frequency-scaling
governor can swamp wall-clock); runs were interleaved across thread
counts with a fixed thermal reset. An independent check uses the
tier~P serial fraction on one chunk ($f=0.0586$, $10$ repetitions,
$95\%$ CI $[0.0553,0.0630]$; wall-clock spread on that instrument
flagged unstable in the deposit --- quote the CI, not a hard point),
measured separately from Work/Span.
That serial-fraction instrument thresholds coefficients at
$\texttt{atol}=10^{-9}$, whereas the library default and the
term-count anchors elsewhere in this manuscript use
$\texttt{atol}=10^{-10}$; on this $N=150$ operator the
looser cut drops $1536$ near-threshold survivors
($90{,}429{,}952$ versus $90{,}431{,}488$). The timing ratios $f$,
the per-unit rates $t_{\text{wht}}$, $t_{\text{coeffs}}$,
$t_{\text{P}}$, and $t_{\text{term}}$, and the drain speedups are
taken from that same instrumented run; the manuscript's term-count
claims follow the $10^{-10}$ deposit. Amdahl's forecast
\citep{amdahl_1967} gives
$1.89\times$ at $p=2$ and $3.40\times$ at $p=4$ --- below Brent's
bounds at the same $p$. A threaded drain over all $5520$ chunks
measures $1.85\times$ and $3.45\times$ ($10$ repetitions, $95\%$ CI
$[1.82,1.87]$ and $[3.40,3.48]$): under the two-worker Amdahl
ceiling, and at four workers within CI of the $3.40\times$ forecast
(point estimate slightly above).

\subsubsection{Memory-bandwidth ceiling on the kernel}

Independent of multi-core drain, how fast can the single-threaded WHT
kernel go on one chunk of the $N=150$ operator
($\dim=16384$)? Compare two platform ceilings: sustained memory
bandwidth versus peak arithmetic throughput. Their ratio --- peak
throughput divided by sustained bandwidth --- is the \emph{ridge
point} (FLOPs/byte). A kernel with arithmetic intensity below that
ridge is memory-bound; above it, compute-bound. The intensity itself is
portable --- fixed by the access pattern --- but the ridge is not, so
the same kernel can sit on different sides of the ridge on different
machines.

On this platform the ceilings are $22.5$~GB/s sustained bandwidth and
$28.8$~GFLOPs/s peak throughput, hence a ridge point of
$1.28$~FLOPs/byte. Bandwidth is from a STREAM-style copy --- one read
and one write per element, matching the butterfly --- over $64$~MiB
arrays (beyond the $8$~MiB L3, so traffic is to DRAM), pinned to one
core. Two implementations, \texttt{numpy.copyto} and a C
\texttt{memcpy}, agreed to within $0.84\%$ ($22.55$ and
$22.36$~GB/s). Peak throughput is counted from the ISA at the
$1.8$~GHz base clock ($4$ doubles per AVX2 vector, $2$ FMA units,
$2$ FLOPs per FMA), not turbo: under sustained AVX2/FMA load the
effective clock fell from $3.04$~GHz after warm-up to a $1.89$~GHz
median, near base rather than the $4.0$~GHz turbo ceiling.

A naive butterfly --- one DRAM pass per stage, no reuse across the
$\log_2\dim$ stages --- has intensity $0.0625$~FLOPs/byte ($1$
FLOP/element against $16$ bytes/element), $20\times$ below the ridge
point. The deployed kernel cache-blocks at an L1-derived
$\mathrm{tile}$ (build-time probe, not a free tuning knob). On
$32$~KiB L1, $\mathrm{tile}=1024$ ($\log_2\mathrm{tile}=10$), so
$10$ stages share one pass and $4$ stages still pay full traffic:
$5$ effective passes instead of $14$, intensity
$0.0625\times(14/5)=0.175$~FLOPs/byte --- still $7.3\times$ below the
ridge, but $2.8\times$ the naive baseline. Even perfect blocking (one
pass for all $14$ stages) reaches only $0.875$~FLOPs/byte,
$1.46\times$ below the ridge. Cache-blocking alone cannot make this
kernel compute-bound here; the deployed tile already takes a
$2.8\times$ step toward that hard $14\times$ ceiling.

The measured
$t_{\text{wht}}=0.749$\,ns per element
(Section~\ref{sec:quantifying-work-and-span}) checks the picture.
At one FLOP per element that is $1.34$~GFLOPs/s --- $4.6\%$ of the
compute peak, where a memory-bound kernel should sit. Dividing by the
tiled intensity recovers $7.64$~GB/s achieved, $34\%$ of the STREAM
ceiling. The shortfall is expected: STREAM is one sequential stream,
while the tiled butterfly re-enters memory on many short tiles across
the full-traffic stages. The check is not circular:
$t_{\text{wht}}$ is from timing the kernel; $22.5$~GB/s comes from a
separate microbenchmark.

\subsection{Memory scaling with result size}
\label{sec:memory-scaling}

Section~\ref{sec:method} asserts that, on the array-yielding path, peak
memory is a per-chunk working set of order
$O(\texttt{chunk\_size}\cdot\dim)$ together with the survivors
$O(n_{\text{terms}})$, rather than a dense accumulate buffer. The
claim concerns \texttt{paulikit}'s own allocations, and each run
therefore reports two meters. \texttt{decompose} tracks byte-level
\texttt{malloc}/\texttt{free} from first touch of the input --- the
quantity addressed by the bound --- while \texttt{peak\_rss} is the
whole-process resident set, including the caller's input matrix, and
is reported so that the end-to-end footprint remains visible rather
than as a proxy for the bound.

On the oscillator family with \texttt{chunk\_size}$=2$ and
$N=50,\ldots,300$, surviving terms grow $1183\times$
($1{,}231{,}872$ to $1{,}457{,}121{,}280$) while
\texttt{decompose} grows $38.2\times$ ($0.889$ to $34.0$~MiB) and
\texttt{peak\_rss} only $1.28\times$ ($121.7$ to $155.4$~MiB) ---
peak resident set flat against three orders of magnitude in result
size (Figure~\ref{fig:memory-scaling}, left). Each point passed the
survivors-only correctness gate (identical term counts across three
repetitions); no repetition exceeded a $1.15\times$ spread threshold.

\begin{figure}[!htbp]
\centering
\includegraphics[width=\linewidth]{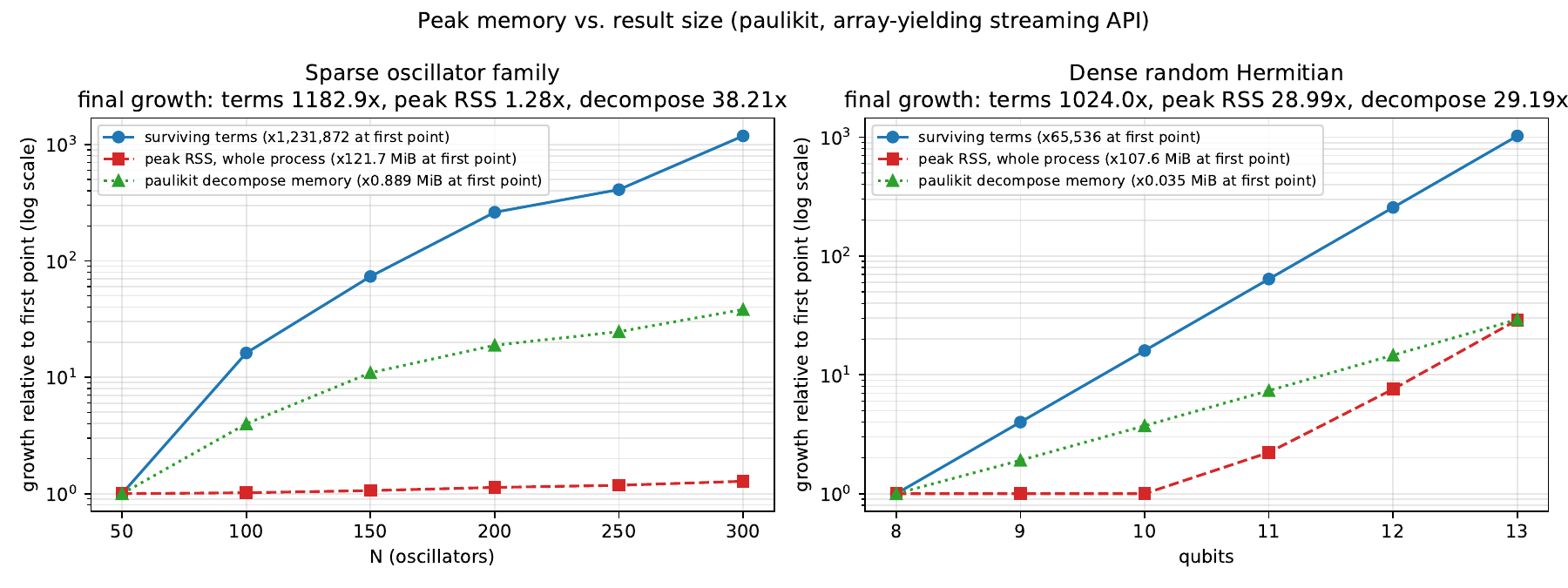}
\caption{Surviving term count, whole-process peak resident memory, and
paulikit's own decomposition memory, all normalized to their value at
the sweep's first point and plotted on a shared log scale, so the gap
between term-count growth and either memory curve's growth is directly
visible. \emph{Left}: sparse oscillator family, $N=50,\ldots,300$,
chunk\_size fixed at $2$ --- terms grow $1183\times$, peak RSS
$1.28\times$, paulikit's own decomposition memory $38.2\times$.
\emph{Right}: dense random Hermitian input with
\texttt{assume\_dense=True}, $n_{\text{qubits}}=8,\ldots,13$ --- terms
grow $1024\times$, peak RSS $29.0\times$ (entirely the cost of holding
the input matrix itself), while paulikit's own decomposition memory
grows only $29.2\times$ off a sub-megabyte base, peaking at
$1.034$~MiB.}
\label{fig:memory-scaling}
\end{figure}

The dense random-Hermitian sweep with \texttt{assume\_dense=True}
($n_{\text{active}}=\dim$, the bound's least favorable case) separates
the two meters sharply (Figure~\ref{fig:memory-scaling}, right). At
$n=8,\ldots,13$, terms grow $1024\times$ ($65{,}536$ to
$67{,}108{,}864$) while \texttt{decompose} grows only $29.2\times$,
peaking at $1.034$~MiB --- a fixed-size per-chunk working set
decoupled from the result. By contrast \texttt{peak\_rss} grows
$29.0\times$ ($107.6$ to $3119$~MiB): a dense $8192\times8192$
\texttt{complex128} matrix is already $1$~GiB, and NumPy's
construction transiently needs more before \texttt{paulikit} runs.
That $\dim^2$ RSS growth is the cost of holding a dense input in
core in this experiment. A non-streaming dense FWHT that still
materializes a full accumulate buffer cannot escape a large working
set by out-of-core input alone; the memory-bound scaling above is a
property of the streaming (chunked) decomposition, not of denseness.
It is not an accumulate-then-threshold working set of size
$O(n_{\text{active}}\times\dim)$ on top of the input --- that extra
cost is absent here.

Without \texttt{assume\_dense=True}, the same dense sweep takes the
general sparsity-detection path (\texttt{np.nonzero} plus
sort-and-scatter). Then \texttt{decompose} grows $833\times$
($8.1$~MiB at $n=8$ to $6722$~MiB at $n=13$), nearly in step with the
$1024\times$ growth in term count --- a real cost, and avoidable when
the caller already knows the input is dense.

Holding $N=300$ and sweeping \texttt{chunk\_size}$\in\{1,2,4\}$
isolates the other half of the bound: \texttt{decompose} rises
linearly ($30.0$, $34.0$, $42.1$~MiB) and \texttt{peak\_rss} with it
($147.0$, $150.9$, $159.7$~MiB), while the term count stays fixed.
The $\texttt{chunk\_size}\times\dim$ term is measured, not only
asserted.

\FloatBarrier
\subsection{Parallel execution model}
\label{sec:process-vs-thread}

Section~\ref{sec:theoretical-ceiling} bounded the parallelism admitted
by the dependency graph. The present subsection records how work
reaches workers on the array-yielding path (Section~\ref{sec:method}).

Two substrates are available. A \emph{process pool}
(\texttt{executor="process"}) gives each worker its own address space
and, with affinity, a physical core, at the price of pickling every
chunk in and every result out. A \emph{threaded drain}
(\texttt{executor="auto"} once both compiled kernels are present)
keeps zero-copy dispatch; multi-core concurrency follows because those
kernels release the GIL in their inner loops.

On top of chunk independence, the threaded drain adds three scheduling
choices. Workers receive a static, capacity-weighted partition up
front --- contiguous slices sized inversely to the number of workers
sharing a physical core, rather than equal chunk counts --- so as to
avoid shared-queue contention and the imbalance of a naive equal split
when the worker count exceeds the core count. Every worker is created
and barrier-synchronized before the first chunk (eager spin-up),
avoiding staggered thread creation on early tasks. When a worker's
share is exhausted, it steals from the fullest victim (the peer with
the most remaining work), correcting runtime imbalance from frequency
scaling or contention that a fixed schedule cannot foresee.

The substrate trade is therefore affinity-plus-pickling versus
zero-copy-plus-GIL. On this platform, with the compiled kernels
present, pickling outweighs the affinity guarantee, and \texttt{auto}
defaults to threads; the crossover is set by this host's inter-process
shipping cost relative to per-chunk compute and may move on another
machine. Capacity weighting, spin-up, and stealing affect only how
evenly independent chunks reach idle workers, not correctness
(Lemma~\ref{lem:chunk-independence}). Whether that evenness holds in
practice is measured next.

\FloatBarrier
\subsection{Parallel efficiency versus single-core baseline}
\label{sec:parallel-scaling}

Measured drain speedup matches the figures already quoted from the
same deposit in Section~\ref{sec:scope-and-limits}: $1.85\times$ at
two workers and $3.45\times$ at four (ten repeats; $95\%$ CIs
$[1.82,1.87]$ and $[3.40,3.48]$), against that run's Amdahl ceilings
$1.89\times$ and $3.40\times$. How closely the run approaches the
ceiling is a test of the scheduling of
Section~\ref{sec:process-vs-thread}, not a second bound on the
algorithm.

Multi-core benches in the same deposit make two further points on this
four-physical-core host. Pushing the threaded drain to eight workers
(the logical/hyperthread count) is slower than four --- wall time
rises --- so SMT oversubscription does not buy throughput here.
Separately, the process-pool executor is slower than single-core
decomposition at every tested worker count; the cost is pickling across
address spaces, not a failure of the threaded drain's
partition-and-steal schedule, which that executor does not use.

Figure~\ref{fig:idle-threads} records the occupancy check: the
kernel-verified active-thread count over a full run at the platform's
physical (four) and logical (eight) core counts. The question is
whether the partitioning-plus-stealing design of
Section~\ref{sec:process-vs-thread} sustains multi-core occupancy
rather than front-loading it. The choice $N=200$, rather than the
manuscript's usual $N=150$, follows from earlier drain-design runs at
$N=200$ that produced the worst tail collapse seen in that
investigation.

\begin{center}
\includegraphics[width=\linewidth]{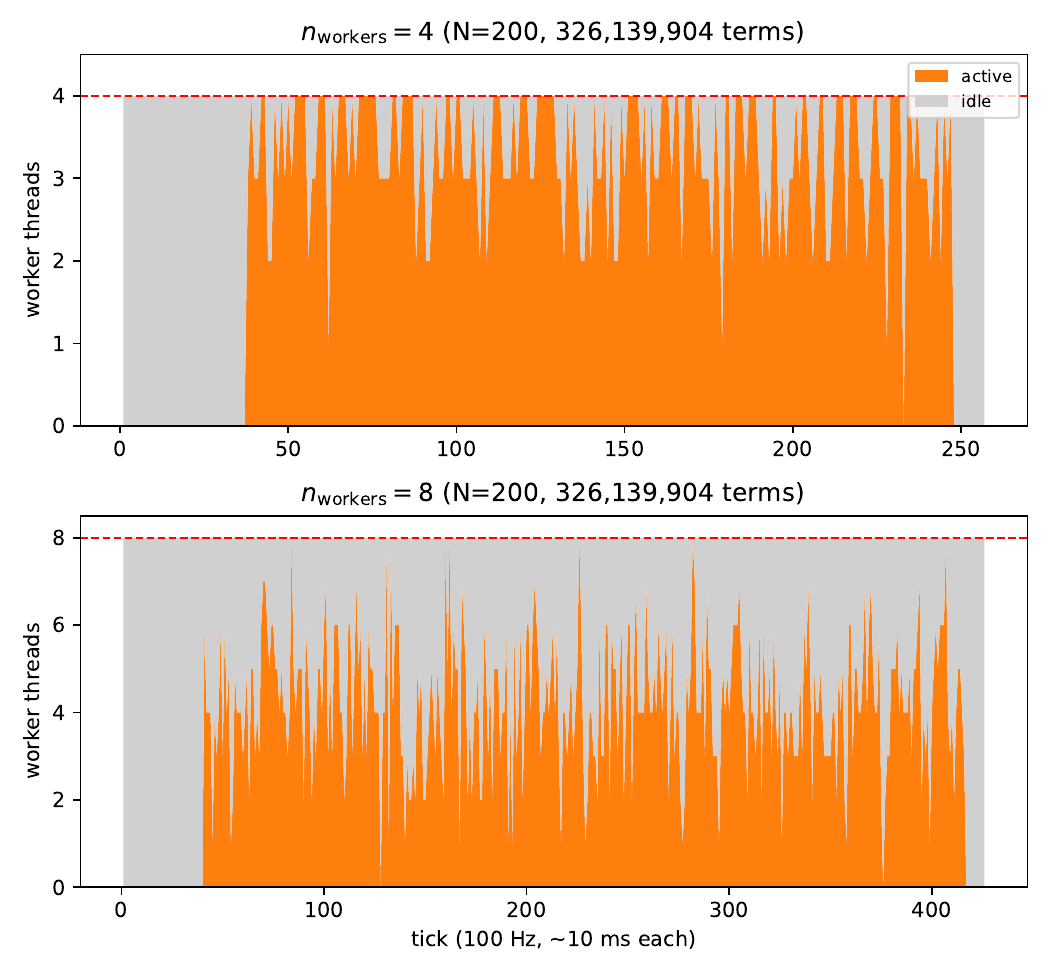}
\captionof{figure}{Active worker-thread count over a full decomposition,
sampled at $100$~Hz via uprobes on the compiled kernels' entry and exit
points (active $=$ inside a kernel, not scheduler state).
\emph{Top}: $n_{\text{workers}}=4$ (physical cores).
\emph{Bottom}: $n_{\text{workers}}=8$ (logical/hyperthreads). Dashed
line: $n_{\text{workers}}$; shaded band: idle gap. Occupancy shape is
a drain-scheduling check (earlier $N=200$ design runs), not a
term-count claim. Python~$3.14$: CPython~$3.13$+ GIL-handoff changes
can affect concurrency on this platform
(Section~\ref{sec:limitations}), so do not compare this pattern
quantitatively to a different CPython without that caveat.}
\label{fig:idle-threads}
\end{center}

\FloatBarrier
\section{Limitations and outlook}
\label{sec:limitations}

Asymptotically, the present work remains in the same $O(n\cdot 4^n)$
FWHT time class as Hamaguchi, Hamada and Yoshioka
\citep{hamaguchi_et_al_2024} and Georges et al.\
\citep{georges_et_al_2025}; it does not advertise a better big-O, nor
the polynomial query complexity of Spencer, Bharti and Gorshkov
\citep{spencer_bharti_gorshkov_2026} under a Pauli-sparsity promise
with sparse query access. The emphasis is the streaming architecture
that makes large \emph{matrix-sparse} runs feasible (peak memory
independent of a full~\mbox{$2^n\times 2^n$} buffer), together with
the character-theoretic derivation (bare versus genuine), open
verification, and performance ceilings checked on one stated
platform.
The numeric ceilings and parallel-efficiency figures are
host-specific (four-physical-core GNU/Linux, stated \texttt{perf}/DVFS
protocol) and specific to the \texttt{paulikit}
\texttt{parallel\_decompose\_arrays} option sets recorded in the
measurements deposit; they do not extend to other hosts or untested
call patterns. Dense runs reported here still assume an in-core input matrix: the
$\dim^2$ growth in \texttt{peak\_rss}
(Section~\ref{sec:memory-scaling}) is precisely that cost. Streaming
already removes the accumulate buffer, not the input; an out-of-core
dense-input path (so the resident set need not hold a full
\mbox{$2^n\times 2^n$} array) exists in later \texttt{paulikit}
($0.2.0$) and is not part of the $0.1.4$ measurement campaign here.
Figure~\ref{fig:idle-threads}
is scheduling evidence under the recorded CPython; $3.13$+
GIL-handoff changes can move threaded-drain occupancy, so do not
treat that occupancy shape as portable across CPython majors. A
dedicated qudit module on Corollary~\ref{cor:qudit-decomposition} is
likewise not shipped; non-Hermitian operators are first-class inputs,
while a worked non-Hermitian physical companion remains follow-on
work.

\section{Conclusion}
\label{sec:conclusion}

Matrix-input Pauli decomposition by FWHT costs $O(n\cdot 4^n)$; the
contribution of this paper is not a better exponent. It consists in a
character-theoretic account of why that transform is the appropriate
one --- labels, sections of the central extension, and the qudit lift
--- together with a streaming realization whose peak memory, for
matrix-sparse inputs, need not densify to a full $2^n\times 2^n$
buffer. Prior fast codes in the same class keep an $O(1)$ additional
working set on a mandatory dense load; that architectural contrast ---
peak memory without densifying a matrix-sparse input --- is the
comparison owned here.

On the stated platform the claim holds through the coupled-oscillator
sweep at $N{=}300$, under Work--Span and memory-bandwidth ceilings
checked without violation.
Dense experiments reported here still keep the input matrix in core;
lifting that $\dim^2$ resident cost by out-of-core ingestion, shipping
the qudit corollary, and writing a non-Hermitian companion case study
are the natural next steps invited by the present construction.

\section*{Data and Code Availability}
\label{sec:availability}

Scripts, raw JSON, and analyses for Section~\ref{sec:results} are
Zenodo \emph{paulikit measurements}~1.0.3,
DOI~\href{https://doi.org/10.5281/zenodo.24030142}{\nolinkurl{10.5281/zenodo.24030142}}
(concept
DOI~\href{https://doi.org/10.5281/zenodo.22992460}{\nolinkurl{10.5281/zenodo.22992460}}).
Reproduce from \texttt{results/COMMANDS.md} inside that deposit's
tarball; protocol and exact-versus-qualitative quantities are in the
deposit \texttt{README.md} and \texttt{METHODOLOGY.md}. Each result
file records the \texttt{paulikit} commit. The characterized library is
\texttt{paulikit}~0.1.4 at
DOI~\href{https://doi.org/10.5281/zenodo.23972284}{\nolinkurl{10.5281/zenodo.23972284}}
(concept
DOI~\href{https://doi.org/10.5281/zenodo.22992251}{\nolinkurl{10.5281/zenodo.22992251}};
GPL-3.0-or-later; also PyPI, Codeberg, GitHub). Oscillator term
counts and memory/publication benches use that release at library
default $\texttt{atol}=10^{-10}$; the serial-fraction / Work--Span
timing deposit used $\texttt{atol}=10^{-9}$ on the same $N=150$
operator (Section~\ref{sec:scope-and-limits}). The Work-versus-wall
cross-check in that section is a host stopwatch against the deposited
rates, not a separate JSON in~$1.0.3$. This manuscript is CC BY~4.0.

\section*{Author Contributions}
\label{sec:contributions}

All authors contributed to the implementation and discussions; M.K.\
and M.M.\ ran the measurements. LLMs assisted with language editing
and with code under review; all results and final wording were checked
and approved by the authors.

\bibliographystyle{unsrtnat}
\bibliography{refs/bibliography}

\end{document}